\documentclass[12pt]{amsart}

\usepackage{amsmath,amsthm,amsfonts,amscd,amssymb}
\usepackage{pst-node}
\usepackage{pst-plot}
\usepackage{fullpage}
\usepackage{tikz, braids}
\usetikzlibrary{decorations.pathreplacing, arrows}

\usetikzlibrary{shapes,snakes,positioning,decorations.pathmorphing,calc,backgrounds,chains,fit,matrix,trees,scopes,shapes.geometric,shapes.multipart}
\tikzset{TL/.style={scale=.5, ineqn}}
\tikzset{JWP/.style={scale=.5, ineqn}}
\tikzset{TLnode/.style = {inner sep = 0, minimum size = 1cm}}
\tikzset{TL at/.style = {shift={($(#1)-(.5,.5)$)}, execute at begin scope = {\draw (0,0) rectangle (1,1);}}}
\tikzset{overstrand/.style={preaction={draw=white, -, line width=6pt}}}
\tikzset{normalizedvert/.style={circle,draw, minimum size = 1mm, inner sep = 0, fill=white}}
\tikzset{coordish/.style={minimum size=0, inner sep = 0, fill=black}}
\tikzset{nicelabel/.style={text height=1.5ex,text depth=.25ex}} % This is used to vertically align different labels
\tikzset{horiznicelabel/.style={text width = .5em}} % horizontally align different labels. used in chapter 8
\tikzset{ineqn/.style={baseline = {($(current bounding box.center)-(0,1ex)$)}}}
\tikzset{ctrl/.style={controls = { #1 and #1} }}
\tikzset{fibabctree/.style={execute at begin picture={  %used afer Fib example in chapter 6
 \draw (-.6,0) -- node[auto,swap]{$a$} (0,0) coordinate(botvert) -- node[auto,swap]{$c$} (.6,0);
 \draw (botvert) -- node[auto,swap]{$b'$} (0,.5) coordinate(topvert);
 \foreach \sign in {+,-}
  \draw (topvert) -- +(\sign 1/6,.5) coordinate(leaf\sign);
}}}
\tikzset{fibacbase/.style={baseline=.5cm,nicelabel, execute at begin picture={  %used afer Fib example in chapter 6
 \draw (-.6,0) node[below right]{$a$} -- (-.25,0) coordinate(leftshadow) -- (0,0) coordinate(botvert) node[below]{$#1$} -- (.25,0) coordinate(rightshadow) -- (.6,0) node[below left]{$c$};
}}}

\tikzset{combspacing/.style = {row sep = .5cm, column sep = 1cm}}  %for combs drawn as matrices, in chap 1 and elsewhere
\tikzset{comb/.style = {combspacing, ampersand replacement = \&, row 1/.style=coordish, matrix of nodes, nodes in empty cells, nodes=draw}}
\usepackage{hyperref}

\title[Highly entangled, non-random subspaces of tensor products from quantum groups]
{Highly entangled, non-random subspaces of tensor products from quantum groups}

\author {Michael Brannan}
\address{Michael Brannan,
Department of Mathematics,
Mailstop 3368, Texas A\&M University, 
College Station, TX 77843-3368, USA}
\email{mbrannan@math.tamu.edu}

\author {Beno\^\i{}t Collins}
\address{Beno\^\i{}t Collins,
Department of Mathematics, Kyoto University,
and CNRS, Institut Camille Jordan Universit\'e  Lyon 1,
France}
\email{collins@math.kyoto-u.ac.jp}

\theoremstyle{plain}
\newtheorem{lemma}{Lemma}[section]
\newtheorem{theorem}[lemma]{Theorem}
\newtheorem{proposition}[lemma]{Proposition}
\newtheorem{corollary}[lemma]{Corollary}

\theoremstyle{definition}
\newtheorem*{definition}{Definition}

\theoremstyle{remark}
\newtheorem{remark}{Remark}

\newcommand{\N}{\mathbb N}
\newcommand{\C}{\mathbb C}

\newcommand{\G}{\mathbb G}
\newcommand{\T}{{\mathbb{T}}}

\newcommand{\mc}{\mathcal}

\begin{document}

\begin{abstract}
In this paper we describe a class of highly entangled subspaces of a tensor product of finite dimensional Hilbert spaces arising from the representation theory of free orthogonal quantum groups.  We determine their largest singular values and obtain lower bounds for the minimum output entropy of the corresponding quantum channels.  
An application  to the construction of $d$-positive maps on matrix algebras is also presented.
\end{abstract}

\maketitle

\section{Introduction}

Entanglement is one of the most important properties that differentiates quantum phenomena from classical phenomena. This property pertains to 
bi-partite or multi-partite systems. In a classical context, a multi-partite system is modeled by a Cartesian product of sets (e.g. of state spaces), whereas in the quantum context, where linear structures are required, 
the Cartesian product is replaced by the tensor product (of Hilbert spaces
describing each individual system).  For example, if $H_A$ and  $H_B$ describe the states of systems $A$ and $B$, then the bipartite system $AB$ is described by the Hilbert space tensor product $H=H_A\otimes H_B$.

Given a Hilbert space $H$, a (pure) state $\xi\in H$ is a vector of norm $1$, taken up to a phase factor. Equivalently, a pure state $\xi$ can be viewed as the rank one projection $\rho_\xi = |\xi \rangle \langle \xi|$ onto $\C \xi \subseteq H$ in
$\mc B(H)$. The (closed) convex hull of pure states is called the state-space of $H$, and denoted by $\mc D(H)$. This is a convex compact set, and its extremal points are the rank one projectors, i.e., pure states.  Given a bipartite system modeled by the Hilbert space tensor product $H = H_A \otimes H_B$, a state $\rho \in \mc D(H)$ is said to {\it separable} if it belongs to the convex hull of the set of product states $\rho = \rho_{A}\otimes \rho_{B}$, where $\rho_{A} \in \mc D(H_A)$ and  $\rho_{B}\in \mc D(H_B)$.   A state $\rho$ is called {\it entangled} if it is not separable.  We shall call a Hilbert subspace $H_0 \subset H_A \otimes H_B$ an {\it entangled subspace} if all of its associated pure states are entangled. 

For the sake of simplicity in this introduction (precise definitions will be given later), 
we say that a Hilbert subspace $H_0\subset H = H_A \otimes H_B$ is {\it highly entangled} if the set of pure states on $H$ associated to $H_0$ are uniformly ``far away'' from the set of product states $\rho_A \otimes \rho_B \in \mc D(H)$. 
A maximally entangled state on $H$ is usually called a {\it Bell state}:  All of the singular values (Schmidt-coefficients) associated to a Bell state are equal. 
It is easy to see that if the dimensions of $H_A$ and $H_B$ are equal, 
the only subspace $H_0 \subseteq H$ such that all its associated pure states are maximally entangled is a dimension one space $H_0 = \C \rho_{B}$ spanned by a Bell state $\rho_{B}$.  

Naturally, the larger the dimension of subspace $H_0 \subseteq H$, the less likely it will be highly entangled, as per the above definition of entanglement.  In recent years it has become a very important problem in Quantum Information Theory (QIT) to: {\it Find subspaces $H_0$ of  large relative dimension in a tensor product $H = H_A \otimes H_B$  such that all states are highly entangled}. 

As mentioned earlier, a possible quantification of entanglement is the distance to separable states. 
Another definition that is widely used is to take the minimum over all states $\xi \in H_0$ of the {\it entanglement entropy} $H(\rho_\xi)$, namely the entropy of the corresponding vectors of singular values (or Schmidt coefficients) associated to $\xi$. 
 This entropy is frequently the Shannon entropy, but it could be other entropies, e.g. the R\'enyi entropies.  We refer to Section \ref{prelim} for precise definitions.
 
Perhaps the most well-known example of a highly entangled vector subspace of large relative dimension is given by the antisymmetric subspace $H\wedge H \subset H\otimes H$ \cite{GrHoPa10}.
One can easily show that each state on $H\wedge H$ has entanglement entropy $\log 2$ and the relative dimension $\frac{\dim H\wedge H}{\dim H \otimes H}$ is approximately equal to $\frac12$. 
In order to make our problem concrete, we could say that we are interested in finding $H_0 \subset H = H_A \otimes H_B$ such that  the numerical quantity 
\begin{align}\label{quantity}\mc E_\mu =  \sup_{\xi \in H_0, \ \|\xi\|=1}\Big\{H(\rho_\xi)  +\mu(\log \dim H_0 -\log \dim H)\Big\} \qquad (0 < \mu < 1)\end{align}
is as large as possible (in particular, positive). 

One rich source of highly entangled subspaces comes from random techniques. The idea of studying random subspaces of tensor products dates back to the work of Hayden, Leung, Shor, Winter, Hastings \cite{HaLeShWi04, HaWi08, HaLeWi06, Ha09}, among others, and it was explored in great detail by Aubrun, Collins, Nechita, Belinschi, Szarek, Werner \cite{AuSzWe11, AuSzYe14, BeCoNe12, FuKi10}, and others.  Unfortunately, these highly random techniques provide no information on finding concrete examples that are predicted to exist by these methods.  Thus, there is a need for a systematic development of non-random examples of highly entangled subspaces.

An attempt was made in this direction by M. Al Nuwairan \cite{Al13,Al14}, by studying the entanglement of subrepresentations of tensor products of irreducible representations of the group $SU(2)$. It is likely that the representation theory of other non-commutative compact matrix groups would give other interesting examples, however $SU(2)$ is arguably {\it the} infinite compact non-commutative group whose representation theory is fully understood, thus allowing for a complete analysis of entanglement and QIT related questions.  In the context of $SU(2)$, Al Nuwairan shows that entanglement always achieved (except when one takes
the highest weight subrepresentation of a tensor product of $SU(2)$-irreducibles).  However, as is evidenced by the results in \cite[Section 3]{Al13}, a high degree of entanglement is unfortunately not achieved when working with $SU(2)$. 

The purpose of the present paper is to initiate the exploration of a new non-random technique to produce highly entangled subspaces.  To some extent, it follows the spirit of works by Al Nuwairan, but it shifts it from the
case of $SU(2)$ and other compact matrix groups to the world of {\it compact quantum groups},
and relies crucially on some geometric ideas surrounding a concept that is well known in the operator algebraic quantum group communitity, namely, the {\it property of rapid decay (RD)} for quantum groups. 

The property of rapid decay for quantum groups was introduced by Vergnioux \cite{Ve07}, and relates to the problem of estimating the $L^\infty$-norms of polynomial functions on quantum groups in terms of their (much easier to calculate) $L^2$-norms.  The operator algebraic notion of property RD has its origins in the groundbreaking work of Haagerup \cite{Ha78} on approximation properties of free group C$^\ast$-algebras.  Unlike in the case of ordinary groups, where property RD is connected to the combinatorial geometry of a discrete group $G$, in the quantum world, property RD was observed by Vergnioux to be intrinsically connected to the geometry of the relative position of a subrepresentation of a tensor product of irreducible  representations of a given quantum group.  More precisely, Vergnioux \cite[Section 4]{Ve07} points out that property RD for a given quantum group $\G$ is related to the following geometric  requirement: {\it Given any pair of irreducible representations $H_A, H_B$ of $\G$,  all multiplicity-free irreducible subrepresentations $H_0 \subset H_A \otimes H_B$ must be asymptotically far from the cone of decomposable tensors in $H_A \otimes H_B$}.

An exploration of this premonitory remark turns out to be extremely fruitful for a certain class of compact quantum groups, called the  {\it free orthogonal quantum groups} $(O_N^+)_{N \ge 3}$.  This remarkable class of quantum groups, introduced by Wang \cite{Wa95}, forms a centerpiece in the theory of C$^\ast$-algebraic compact quantum groups.  $O_N^+$ arises as a certain universal non-commutative deformation of the  function algebra on the classical real orthogonal group $O_N$, and has been the topic of much study over the past 20 years.  See, for example, the survey \cite{Br16} and the references therein.  One remarkable fact for our purposes, discovered by Banica \cite{Ba96}, is that the quantum groups $O_N^+$ have a unitary representation theory that closely parallels that of $SU(2)$.  In particular, the unitary irreducible representations of $O_N^+$ have the same fusion rules as $SU(2)$, and their construction is well understood in terms of the planar calculus of the Temperley-Lieb category \cite{KaLi94}.   This close parallel with $SU(2)$, on the one hand, allows for a highly computable framework (like one has for $SU(2)$).  On the other hand, the genuinely quantum  features of $O_N^+$ result in a much higher degree of entanglement in subrepresentations of tensor products, in comparison to what can be obtained for $SU(2)$.

For the free orthogonal quantum groups $(O_N^+)_{N \ge 3}$, we show that one can describe very precisely the largest singular values of states that appear in irreducible subrepresentations of tensor product representations (see Theorem \ref{ent-saturation}).  As a result we describe very precisely a new non-random class of subspaces of tensor products with the property of being highly 
entangled and of large relative dimension.  In particular we find deterministic examples of entangled subspaces of large relative dimension such that the quantity $\mc E_\mu$ defined in \eqref{quantity} is strictly positive for any $\mu < 1/2$.     We also deduce from our entanglement results some interesting properties for the class of quantum channels associated to these subspaces. We compute explicitly the $\mc S^1 \to \mc S^\infty$ norms of these channels, and obtain large lower bounds on their minimum output entropies (see Section \ref{app}).     

It is our hope that this paper will be a first step towards substantiating the claim that quantum groups form a rich well of entangled subspaces and quantum channels with interesting analytic properties.  
Along the way, we revisit the fundamental geometric inequality associated to the rapid decay property for $O_N^+$ (Proposition \ref{thm:RD}), and improve our understanding thereof. In particular, we show that entanglement inequality for property RD is essentially optimal for the free orthogonal quantum groups, and establish a higher-rank  generalization of it (Theorem \ref{ent-saturation}).

The remainder of our paper is organized as follows:  After this introduction, we recall in the first part of Section \ref{prelim} some concepts related to entangled subspaces, quantum channels, and minimum output entropy  of quantum channels.  The second half of Section \ref{prelim} introduces the free orthogonal quantum groups and describes aspects of their irreducible unitary representation theory that will be used in the sequel.  The main section of the paper is Section \ref{ent-subspace} where we study the entanglement of irreducible subrepresentations of  tensor products of $O_N^+$-representations.  There we prove the rapid decay inequality (Proposition \ref{thm:RD}) in the spirit of Vergnioux, establishing high entanglement for the subspaces under consideration (Theorem \ref{highest-weight}).  We then go on to strengthen this rapid decay inequality to a higher rank version (Theorem \ref{ent-saturation}) and discuss its optimality. 

Finally, we apply this strengthened rapid decay property in Section \ref{ap} to study the  quantum channels that are naturally associated to our entangled subspaces.  Here we obtain lower bounds for the MOE's of these channels (Corollary \ref{MOE-bounds}) and discuss their sharpness. 
In Section \ref{positive}, we use our entangled subspaces to construct new deterministic examples of $d$-positive maps between matrix algebras.

\subsection*{Acknowledgements}  The first author is grateful to Kyoto University and the University of Ottawa where parts of this work were completed during visits there.  The first author is also indebted to Jason Crann for stimulating communications on topics related to this paper.  B.C. was supported by NSERC discovery and accelerator grants, JSPS Kakenhi wakate B, and ANR-
14-CE25-0003.  

\section{Preliminaries} \label{prelim}

\subsection{Entangled subspaces of a tensor product}

Consider a pair of finite-dimensional complex Hilbert spaces  $H_A$ and $H_B$.  Any unit vector $\xi$ belonging to the tensor product Hilbert space $H_A \otimes H_B$ admits a {\it singular value decomposition}:   There are unique constants $\lambda_1 \ge \lambda_2 \ge \ldots\lambda_{d} \ge 0$ (with $d= \min\{\dim H_A, \dim H_B\}$) and  orthonormal systems $(e_i)_{i=1}^d \subset H_A$ and $(f_i)_{i=1}^d \subset H_B$ such that 
\[
\xi = \sum_{i=1}^d \sqrt{\lambda_i}e_i \otimes f_i.
\]  The sequence of numbers $(\lambda_i)_i$ is uniquely determined (as a multi-set) by $\xi$ and these numbers are called the {\it singular values (or Schmidt coefficients)} of $\xi$.  Note that $\|\xi\|^2 = \sum_{i=1}^d \lambda_i$.  

We shall call a non-zero vector $\xi \in H_A \otimes H_B$ {\it separable} if there exist vectors $\eta \in H_A$, $\zeta \in H_B$ such that $\xi = \eta \otimes \zeta$.  If $\xi$ is not separable, it is called {\it entangled}.  Note that a unit vector $\xi \in  H_A \otimes H_B$ is separable if and only if its corresponding sequence of Schmidt coefficients is $(1, 0, 0 , \ldots, 0)$.   We shall similarly call a  linear subspace $H_0 \subseteq H_A \otimes H_B$ separable (resp. entangled) if $H_0$ contains (resp. does not contain) separable  vectors.   We note that the maximally entangled unit vector $H_A \otimes H_B$ is the so-called  {\it Bell Vector (Bell state)} $\xi_B$, whose singular value decomposition is given by
\[
\xi_B =  \frac{1}{\sqrt{d}}\sum_{i=1}^d e_i \otimes f_i \qquad (d \ge 2).  
\] 

Note that the Schmidt coefficients of the Bell vector are given by $\Big(  \frac{1}{\sqrt{d}},  \frac{1}{\sqrt{d}}, \ldots,  \frac{1}{\sqrt{d}}, 0,0, \ldots\Big)$.  In particular, the largest Schmidt coefficient $\lambda_1$ associated to a unit vector $\xi \in H_A \otimes H_B$ is maximized at $1$ precisely when it is separable, and it is minimized at $d^{-1/2}$ when $\xi = \xi_B$ is the Bell state.  In this sense, the singular value decomposition is a useful tool for measuring measure how entangled a unit vector $\xi \in H_A \otimes H_B$ is: If $\lambda_1 << 1$, then $\xi$ is highly entangled.  With this in mind, we call a linear subspace $H_0 \subseteq H_A \otimes H_B$ {\it highly entangled} if the supremum of all maximal Schmidt coefficients associated to all unit vectors in $H_0$ is bounded away from one.   That is, 
\[
\sup_{\xi \in H_0, \ \|\xi\| = 1} \lambda_1 < < 1.  
\]
Equivalently, $H_0\subseteq H_A \otimes H_B$ is highly entangled if and only if 
\begin{align}\label{entangled}
\sup_{\|\xi\|_{H_0} = \|\eta\|_{H_A} = \|\zeta\|_{H_B} = 1}|\langle \xi| \eta \otimes \zeta\rangle| <<1.
\end{align}

\subsection{Quantum channels}

Given a finite dimensional Hilbert space $H$, denote by $\mc B(H)$ the C$^\ast$-algebra of bounded linear operators on $H$, and denote by $\mc D(H)\subseteq \mc B(H)$ the collection of {\it states on $H$}: positive semidefinite matrices $0 \le \rho \in \mc B(H)$ satisfying $\text{Tr}(\rho) = 1$, where $\text{Tr}$ denotes the canonical trace on $\mc B(H)$.  A state $\rho \in \mc D(H)$ is called a {\it pure state} if there exists a unit vector $\xi \in H$ so that $\rho$ is given by the rank-one projector $\rho_\xi  = |\xi \rangle \langle \xi |$.  We denote by $\mc S_1(H) \subseteq \mc B(H)$ the linear span of $\mc D(H)$, which is a Banach algebra with respect to the trace norm $\|\rho\|_{\mc S_1(H)} = \text{Tr}(|\rho|)$.    

Given two (finite-dimensional) Hilbert spaces $H_A$ and $H_B$, a {\it quantum channel} is a linear, completely positive and trace-preserving map (CPTP map) $\Phi:\mc B(H_A) \to \mc B(H_B)$ \cite{NiCh}.  By definition, we have $\Phi(\mc D(H_A)) \subseteq \mc D(H_B)$ for any quantum channel $\Phi$.  A natural model for the construction of quantum channels comes from subspaces of Hilbert space tensor products.  Given a triple of finite dimensional Hilbert spaces $(H_A, H_B, H_C)$ and an isometric linear map $\alpha_A^{B,C}:H_A \to H_B \otimes H_C$, we can form a complementary pair of quantum channels 
\begin{align*}
&\Phi_A^{\overline B, C}:\mc B(H_A) \to \mc B(H_C); \quad \Phi_A^{\overline B, C}(\rho) = (\text{Tr}_{H_B} \otimes \iota)(\alpha_A^{ B, C}\rho (\alpha_A^{ B, C})^*) \\
&\Phi_A^{ B, \overline C}:\mc B(H_A) \to \mc B(H_B); \quad \Phi_A^{ B,  \overline C}(\rho) = (\iota \otimes \text{Tr}_{H_C})(\alpha_A^{ B, C}\rho (\alpha_A^{B, C})^*). 
\end{align*}
Remarkably, every quantum channel can be expressed in the above form, thanks to the well known Stinespring Dilation Theorem for completely positive maps.  In other words, given any quantum channel $\Phi:\mc B(H_A) \to \mc B(H_B)$, the Stinespring Theorem guarantees the existence of an essentially unique {\it Stinespring pair} $(H_C, \alpha_A^{B,C})$, where $H_C$ is an auxiliary ``environment'' Hilbert space and $\alpha_A^{B,C}:H_A \to H_B \otimes H_C$ is a linear isometry, so that $\Phi = \Phi_A^{B, \overline C}$ in the above notation.  See \cite{HaWi08}, for example.

The {\it minimum output entropy (MOE)} of a quantum channel $\Phi:\mc B(H_A) \to \mc B(H_B)$ is given by
\[ 
H_{\min}(\Phi) := \min_{\rho \in \mc D(H_A)} H(\Phi(\rho)),
\]     
where $H(\cdot)$ denotes the {\it von Neumann entropy} of a state: $H(\rho) = -\text{Tr}(\rho\log \rho)$.  Note that by functional calculus, we have $H(\rho) = -\sum_i \lambda_i \log \lambda_i$, where $(\lambda_i)_{i} \subset [0,\infty)$ denotes the spectrum of $\rho$.   In other words, $H(\rho)$ is nothing but the Shannon entropy of the probability vector $(\lambda_i)_i$ corresponding to the eigenvalues of $\rho$. 

Since the von Neumann entropy functional $H(\cdot)$ is well-known to be convex, it follows that the MOE $H_{\min}(\Phi)$ is minimized on the extreme points of the compact convex set $\mc D(H_A)$, which corresponds to the set of all pure states on $H$.  In particular, 
\[
H_{\min}(\Phi) = \min_{\xi \in H_A, \ \|\xi\| = 1} H(\Phi( |\xi \rangle \langle \xi |)).
\] 
Using this fact together with the Stinespring Theorem, it can be shown that the $H_{\min}(\Phi)$ depends only on the relative position of the subspace $\alpha_A^{ B, C}(H_A) \subseteq H_B \otimes H_C$ coming from the Stinespring representation $\Phi = \Phi_A^{ B, \overline C}= (\iota \otimes \text{Tr}_{H_C})(\alpha_A^{ B, C}(\cdot) (\alpha_A^{B, C})^*)$.  Indeed, in this case, we have \begin{align*}
H_{\min}(\Phi) &= \min_{\xi \in H_A, \ \|\xi\| = 1} H(|\xi \rangle \langle \xi | ) = \min_{\xi \in H_A, \ \|\xi\| = 1} H((\iota \otimes \text{Tr}_C)(|\alpha_A^{ B, C}(\xi) \rangle \langle \alpha_A^{ B, C}(\xi) | ))  \\
&=  \min_{\xi \in H_A, \ \|\xi\| = 1} -\sum_i \lambda_i \log \lambda_i,
\end{align*}
where $ (\lambda_i)_i$ are the Schmidt coefficients of $\alpha_A^{ B, C}(\xi) = \sum_i \sqrt{\lambda_i} e_i \otimes f_i \in H_B \otimes H_C$.   In particular, $H_{\min}(\Phi)$ is zero if and only if $\alpha_A^{ B, C}(H_A) \subseteq H_B \otimes H_C$ is a separable subspace.

\subsection{Free orthogonal quantum groups and their representations}

In this section we give a very light overview of the free orthogonal quantum groups and some aspects of their finite dimensional representation theory.  Much of what we state below about quantum groups and their representations can be phrased in more general terms, however this will not be needed for our purpose.   The interested reader may refer to \cite{Ti08, Wo98}.

The main idea behind the concept of a free orthogonal quantum group is to formulate a  {\it non-commutative} version of the commutative $\ast$-algebra of complex-valued polynomial functions on the real orthogonal group $O_N$.  
It turns out that if one formulates such a non-commutative $\ast$-algebra in the right way, many of the nice group theoretic structures associated to that fact that $O_N$ is a compact group persist (e.g., a unique ``Haar measure'', a rich finite-dimensional unitary representation theory, a Peter-Weyl theorem, and so on).

\begin{definition}[Free Orthogonal Quantum Groups]
Let $N \ge 2$, let $A$ be a unital $\ast$-algebra over $\C$,and let $u = [u_{ij}]_{1 \le i,j \le N} \in M_N(A)$ be a matrix with entries in $A$.  Write $u^* = [u_{ji}^*] \in M_N(A)$ and $\bar u = [u_{ij}^*] \in M_N(A)$.  
\begin{enumerate}
\item The matrix $u$ is called a {\it quantum orthogonal matrix} if $u$ is invertible in $M_N(A)$, $u^* = u^{-1}$, and $\bar u = u$ 
\item The {\it free orthogonal quantum group} (of rank $N$) is given by the triple $O_N^+:=(\mc O(O_N^+), u, \Delta)$, where
\begin{enumerate}
\item  $\mc O(O_N^+)$ is the universal unital $\ast$-algebra generated by the coefficients $(u_{ij})_{1 \le i,j \le N}$  of a quantum orthogonal matrix $u = [u_{ij}] \in M_N(\mc O(O_N^+))$.
\item $\Delta:\mc O(O_N^+) \to \mc O(O_N^+)\otimes \mc O(O_N^+)$ is the unique unital $\ast$-algebra homomorphism, called the {\it co-product}, given by 
\[
\Delta(u_{ij}) = \sum_{k=1}^N u_{ik} \otimes u_{kj} \qquad (1 \le i,j \le N).
\] 
\end{enumerate}
\end{enumerate}
\end{definition}

\begin{remark}
In the above definition, we more precisely mean that $\mc O(O_N^+)$ is defined as the quotient $\ast$-algebra \[\mc O(O_N^+) := \C\langle X_{ij}, X_{ij}^* | 1 \le i,j \le N \rangle/ \mc I_N,\] where $\mc I_N \lhd  \C\langle X_{ij}| 1 \le i,j \le N \rangle$ is the ideal generated by the relations \[X_{ij} = X_{ij}^*, \quad \sum_k X_{ik}X_{jk}  = \sum_{k} X_{ki}X_{kj} = \delta_{i,j}1.\]  The isomorphism being given by $u_{ij} \mapsto X_{ij} + \mc I_N$.
\end{remark}

\begin{remark}
If we denote by $\mc J_N \lhd \mc O(O_N^+)$ the ideal generaated by all the commutators $[u_{ij}, u_{kl}]$, we obtain the abelianization of $\mc O(O_N^+)/ \mc J_N$  of $\mc O(O_N^+)$, which is isomorphic to $\mc O(O_N^+)$, the $\ast$-algebra of polynomial functions on the real orthogonal group $O_N$.  This identification is given by the map $\mc O(O_N^+) \owns u_{ij} + \mc J_N  \mapsto v_{ij} \in \mc O(O_N)$,  where $v = [v_{ij}] \in M_N(\mc O(O_N))$ forms the matrix of basic coordinate functions on $O_N$.  In this context, the co-product map $\Delta$ on $\mc O(O_N^+)$ factors through the quotient and induces a corresponding map $\Delta$ on $\mc O(O_N)$.  At this level, it is easily seen that $\Delta (f) (s,t) = f(st)$ for all $f \in \mc O(O_N)$ and $s,t \in O_N$.  That is, $\Delta$ reflects the group law on $\mc O(O_N)$ at the level of the function algebra $\mc O(O_N)$.  In this sense, we are justified in calling the quantum group $O_N^+$ a ``free analogue'' of the classical orthogonal group $O_N$.  
\end{remark}

We now turn to the concept of a representation of $O_N^+$.  A (finite-dimensional unitary) {\it representation} of $O_N$ is given by a finite dimensional Hilbert space $H_v$ and unitary matrix $v \in \mc O(O_N^+) \otimes \mc B(H_v)$ satisfying 
\[(\Delta \otimes \iota) v = v_{13}v_{23} \in \mc O(O_N^+) \otimes \mc O(O_N^+)  \otimes \mc B(H_v) ,\]
where above we use the standard leg numbering notation for linear maps on tensor products.  If we fix an orthonormal basis $(e_i)_{i=1}^d \subset H_u$, then we can write $v$ as the matrix $[v_{ij}] \in M_d(\mc O(O_N^+))$ with respect to this basis, and the above formula translates to \[\Delta v_{ij} =  \sum_{k=1}^d v_{ik} \otimes v_{kj} \qquad (1 \le  i,j \le d).\]  The first examples of representations of $O_N^+$ that come to mind are the one-dimensional {\it trivial representation} (which corresponds to the unit $1 \in \mc O(O_N^+) = M_1(\mc O(O_N^+))$) and the  $N$-dimensional {\it fundamental representation} $u = [u_{ij}] \in  M_N(\mc O(O_N^+))$ (corresponding to the matrix of generators for $\mc O(O_N^+)$).  Given two representations $v= [v_{ij}]$ and $w = [w_{kl}]$, we can naturally form their {\it direct sum} $v \oplus w \in \mc O(O_N^+)\otimes \mc B(H_v \oplus  H_w)$ and their {\it tensor product} $v \otimes w = v_{12}w_{13} = [v_{ij}w_{kl}]  \in \mc O(O_N^+)\otimes \mc B(H_v \otimes H_w)$ to obtain new examples of representations from old ones.  From a unitary representation $v = [v_{ij}]$, we may also form the {\it contragredient representation} $\bar v:= [v_{ij}^*] \in \mc O(O_N^+) \otimes \mc B(\overline{H_v})$.  

In order study the structure of various representations of $O_N^+$, we use the concept of intertwiner spaces.  Given two representations $u$ and $v$ of $O_N^+$, define the space of {\it intertwiners} between $u$ and $v$ as 
\[\text{Hom}(u,v) = \{T \in \mc B(H_u, H_v): (\iota \otimes T)u = v(\iota \otimes T).\}\] 
Two representations $u,v$ are called {\it equivalent} if $\text{Hom}(u,v)$ contains an invertible operator, and a representation $u$ is called {\it irreducible} if $\text{Hom}(u,u) = \C 1$.  It is known that every unitary representation of $O_N^+$ is equivalent to a direct sum of irreducible unitary representations. 

It is known from \cite{Ba96} that the irreducible corepresentations of $O_N^+$ can be labelled  $(v^k)_{k \in \N_0}$ (up to unitary equivalence) in such a way that $v^0 = 1$, $v^1 = u$ (the fundamental representation), and the following fusion rules hold:

\begin{align} \label{frules}
v^l \otimes v^m \cong v^{|l-m|} \oplus v^{|l-m|+2} \oplus \ldots \oplus v^{l+m} = \bigoplus_{0 \le r \le \min\{k,l\}} v^{l+m - 2r}.
\end{align}
Moreover the contragredient $\overline{v^k}$ of $v^k$ is unitarily equivalent to  $v^k$ for all $k$.  Denote by $H_k$ the Hilbert space associated to $v^k$.  Then $H_0 = \C$, $H_1 = \C^N$, and the dimensions $\dim H_k$ satisfy the recursion relations 
$\dim H_1 \dim H_k = \dim H_{k+1} + \dim H_{k-1}$.  Defining the quantum parameter \[q = q(N) := \frac{1}{N}\Big(\frac{2}{1+ \sqrt{1 -4/N^2}}\Big) \in (0,1],\] one can inductively show that the dimensions $\dim H_k$ are given by the {\it quantum integers} \[\dim H_k = [k+1]_q: = q^{-k}\Big(\frac{1-q^{2k+2}}{1-q^2}\Big) \qquad (N \ge 3).
 \]
When $N=2$, we have $q=1$, and then $\dim H_k = k+1 = \lim_{q \to 1^-} [k+1]_q$.  Note that for $N \ge 3$, we have the exponential growth asymptotic $[k+1]_q \sim N^k$.  For our purposes, this exponential growth is crucial and therefore we generally assume $N \ge 3$ in the sequel.

The explicit construction of the irreducible representation spaces $(H_k)_{k \in \N_0}$ proceeds as follows \cite{Ba96, VaVe07, BiDeVa06}.  Denote by $\{e_1, \ldots, e_N\} \subset H_1 = \C^N$ the standard orthonormal basis, and put $T_1 = \sum_{i=1}^N e_i \otimes e_i \in H_1 ^{\otimes 2}$.  Then one can readily check that $T_1 \in \text{Hom}(1,u\otimes u)$.  (I.e., $u^{\otimes 2}(1 \otimes T_1) = (1 \otimes T_1)$.)   Next, we consider the intertwiner space $\text{Hom}(u^{\otimes k}, u^{\otimes k}) \subseteq \mc B((\C^N)^{\otimes k})$, which can be shown to contain a unique non-zero self-adjoint projection $p_k$ ({\it the Jones-Wenzl projection}) with the defining property that 
\[
(\iota_{H_1^{\otimes i-1}} \otimes T_1 T_1^*\otimes  \iota_{H_1^{\otimes k-i-1}})p_k= 0 \qquad (1 \le i \le k-1).
\]
One then can identify $H_k$ with the ``highest weight'' subspace $p_k(H_1^{\otimes k}) \subseteq H_1^{\otimes k}$.  The existence and choice of terminology for $p_k$ comes from the fact that the intertwiner space $\text{Hom}(u^{\otimes k}, u^{\otimes k})$ is in fact isomorphic to 
a Temperley-Lieb algebra \cite{TeLi71}, and $p_k$ corresponds precisely to the Jones-Wenzl projection \cite{We87} under this isomorphism.  In particular, the {\it Wenzl recursion}
\[
p_1 = \iota_{H_1}, \quad p_k = \iota_{H_1} \otimes p_{k-1} - \frac{[k-1]_q}{[k]_q}(\iota_{H_1}\otimes p_{k-1})(T_1T_1^* \otimes \iota_{H_1^{\otimes k-2}})(\iota_{H_1}\otimes p_{k-1}) \qquad (k \ge 2)
\]
can be used to determine $p_k$.  In passing, we point out that the problem of obtaining explicit formulas for Jones-Wenzl projections (beyond the above recursion) has attracted a lot of attention over the years from various mathematical communities.  See \cite{BrCo16, Mo15, FrKh97} and the references therein.

We conclude this section with a description of the non-empty intertwiner spaces $\text{Hom}(v^k, v^l \otimes v^m)$ that arise from the fusion rules \eqref{frules}.  To begin, let us call a triple $(k,l,m) \in \N_0^3$ {\it admissible} if there exists an integer $0 \le r \le \min\{l,m\}$ such that $k = l+m - 2r$.  In other words, $(k,l,m) \in \N_0^3$ is admissible if and only if the tensor product representation $v^l \otimes v^m$ contains a (multiplicity-free) subrepresentation equivalent to $v^k$.  Fix an admissible triple $(k,l,m) \in \N_0^3$.  Then $\text{Hom}(v^k, v^l \otimes v^m) \subseteq \mc B(H_k, H_l \otimes H_m) \subseteq \mc B(H_1^{\otimes k}, H_1^{\otimes l} \otimes H_1^{\otimes m})$ is one-dimensional and is spanned by the following canonical non-zero intertwiner
\begin{align}\label{unnormal}
A_k^{l,m} =(p_l \otimes p_m)\Big(\iota_{H_{l-r}} \otimes T_r  \otimes \iota_{m-r}\Big)p_k ,
\end{align} 
where $T_r \in \text{Hom}(1, u^{\otimes 2r})$ is defined recursively from $T_1 = \sum_{i=1}^Ne_i \otimes e_i$ via 
$T_r = (\iota_{H_1} \otimes T_1 \otimes \iota_{H_1})T_{r-1}$.  The maps $A_{k}^{l,m}$ are well studied in the Temperley-Lieb recoupling theory \cite{KaLi94}, and are known there as {\it three-vertices}.  A three-vertex is typically diagrammatically represented as follows:
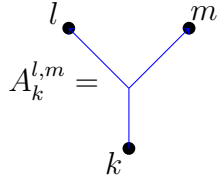
\begin{center}

	\begin{tikzpicture}[baseline=(current  bounding  box.center),
			wh/.style={circle,draw=black,thick,inner sep=.5mm},
			bl/.style={circle,draw=black,fill=black,thick,inner sep=.5mm}, scale = 0.2]
		%%%%%%%%%%%%%%%%%%%%%%%%%%%%%%%%%%%%%%%%%%%%%%%%%%%%%%%%%%%%%%%%%%%%%%
 \node(a) at (-4,4) [bl] {};
%\node(b) at (0,0) [bl] {};
\node(c) at (4,4) [bl] {};
\node(d) at (0,-4) [bl] {};

\node  at (-5,0.5) {$A_{k}^{l,m} = $};	
\node  at (-5,5) {$l$};	
\node  at (5,5) {$m$};	
\node  at (-1,-5) {$k$};	

\draw [<-, color=blue]
		(-4,4) -- (0,0);

\draw [->, color=blue]
		(0,-0) -- (4,4);

\draw [-, color=blue]
		(0,-4) -- (0,0);

	\end{tikzpicture}

\end{center}

In order to find the unique $O_N^+$ equivariant isometry $\alpha_k^{l,m}:H_k \to H_l \otimes H_m$ (up to multiplication by $\T$), we compute the norm of $A_k^{l,m}$.  To do this, we define (following the terminology and diagrammatics from \cite{KaLi94}) the {\it $\theta$-net}
\[\theta_q(k,l,m) = \text{Tr}_{H_k}((A_k^{l,m})^*A_k^{l,m}) = \begin{tikzpicture}[baseline=(current  bounding  box.center),
			wh/.style={circle,draw=black,thick,inner sep=.5mm},
			bl/.style={circle,draw=black,fill=black,thick,inner sep=.5mm}, scale = 0.2]
		%%%%%%%%%%%%%%%%%%%%%%%%%%%%%%%%%%%%%%%%%%%%%%%%%%%%%%%%%%%%%%%%%%%%%%
 \node(a) at (-4,4) [bl] {};
%\node(b) at (0,0) [bl] {};
\node(c) at (4,4) [bl] {};

\node(d) at (0,-4) [bl] {};
\node at (0,12) [bl] {};

%\node  at (-5,0.5) {$A_{k}^{l,m} = $};	
\node  at (-5,5) {$l$};	
\node  at (5,5) {$m$};	
\node  at (-1,-5) {$k$};	
\node at (-1,13) {$k$};

\draw [<-, color=blue]
		(-4,4) -- (0,0);

\draw [->, color=blue]
		(0,-0) -- (4,4);

\draw [-, color=blue]
		(0,-4) -- (0,0);
		
\draw [-, color=blue]
		(-4,4) -- (0,8);
		
\draw [-, color=blue]
		(4,4) -- (0,8);
		
\draw [-, color=blue]
		(0,8) -- (0,12);		

\draw [-, color=blue]
		(0,12) to [bend right] (0,-4);
	\end{tikzpicture}.\] 
Since $A_k^{l,m}$ is a multiple of an isometry, it easily follows that $\|A_k^{l,m}\|^2[k+1]_q = \theta_q(k,l,m)$.  $\theta$-net evaluations are well known \cite{KaLi94, Ve05, VaVe07}, and are given by  
\begin{align}\label{theta}
\theta_q(k,l,m):=\frac{[r]_q![l-r]_q![m-r]_q![k+r+1]_q!}{[l]_q![m]_q![k]_q!},
\end{align}   
where $k=l+m - 2r$ and $[x]_q! = [x]_q[s-1]_q \ldots [2]_q[1]_q$ denotes the quantum factorial.  We thus arrive at the following formula for our isometry $\alpha_k^{l,m}$:
\begin{align}\label{alpha}
\alpha_k^{l,m} = \|A_k^{l,m}\|^{-1}A_k^{l,m} = \Big(\frac{[k+1]_q}{\theta_q(k,l,m)}\Big)^{1/2} A_k^{l,m}.
\end{align}

\begin{remark} \label{upbound}
We conclude this section by observing the following useful inequality for $\theta$-nets: 
\[
\frac{[r+1]_q[k+1]_q}{\theta_q(k,l,m)} \ge 1,
\]
which holds for any admissible triple $(k,l,m) \in \N_0^3$ with $k = l+m -2r$.  This follows from the observation that for a three-vertex $A_k^{l,m}$,
\[
A_k^{l,m} =(p_l \otimes p_m)\Big(\iota_{H_{l-r}} \otimes T_r  \otimes \iota_{m-r}\Big)p_k = (p_l \otimes p_m)\Big(\iota_{H_{l-r}} \otimes (p_r \otimes p_r)T_r  \otimes \iota_{m-r}\Big)p_k,
\] 
\[\implies \|A_k^{l,m}\|^2 \le  \|(p_r \otimes p_r)T_r\|^2  =\text{Tr}(p_r)  = \dim H_r =  [r+1]_q. 
\]
\end{remark}

\section{Entanglement for subrepresentations of $O_N^+$ and rapid decay inequalities}
\label{ent-subspace}

In this section we begin our study of the entanglement geometry of irreducible subrepresentations of tensor products of irreducible representations of $O_N^+$.  The general setup we will consider is a fixed $N \ge 3$ and an admissible triple $(k,l,m) \in \N_0^3$.  This corresponds to irreducible representations $(v^k,v^l,v^m)$ of $O_N^+$ with corresponding representation Hilbert spaces $(H_k,H_l,H_m)$, and a $O_N^+$-equivariant isometry $\alpha_k^{l,m}:H_k \to H_l \otimes H_m$ as constructed in the previous section.  Recall that we set $q = \frac{1}{N}\Big(\frac{2}{1+ \sqrt{1 -4/N^2}}\Big) \in (0,1)$.  Our main interest is to study the entanglement of the $\alpha_k^{l,m}(H_k) \subseteq H_l \otimes H_m$, and the following proposition gives a measure of this.

\begin{proposition} \label{thm:RD}  Fix $N \ge 3$ and let $(k,l,m) \in \N_0^3$ be an admissible triple.  Then
for any unit vectors $\xi \in H_k, \eta \in H_l, \zeta \in H_m$, we have
\begin{align*}
|\langle \alpha_k^{l,m}(\xi)|\eta \otimes \zeta \rangle| \le \Big(\frac{[k+1]_{q}}{\theta_{q}(k,l,m)}\Big)^{1/2} \le   C(q) q^{\frac{l+m-k}{4}}, 
\end{align*}
where 
\[
C(q) =  (1-q^2)^{-1/2}\Big(\prod_{s=1}^\infty \frac{1}{1-q^{2s}}\Big)^{3/2}
\]
\end{proposition}

\begin{remark}
We note that the bound $C(q)q^{\frac{l+m-k}{4}}$ appearing in Proposition \ref{thm:RD} is equivalent, 
as $N$ is large,
to the fourth root of the relative dimension, $\Big(\frac{\dim H_k}{\dim H_l \dim H_m}\Big)^{1/4}$.
\end{remark}

\begin{proof}[Proof of Proposition \ref{thm:RD}]  Fix $\eta \in H_l$ and $\zeta \in H_m$.  We identify $H_l$ with the highest weight subspace of $H_{l-r} \otimes H_1^{\otimes r}$ and similarly we identify $H_m$ with the highest weight subspace of $H_1^{\otimes r} \otimes H_{m-r}$.  I.e, $H_l = p_l(H_{k-r} \otimes H_1^{\otimes r})$ and $H_m = p_m(H_1^{\otimes r} \otimes H_{m-r})$, where $p_l, p_m$ are the corresponding Jones-Wenzl projections.  Then we can uniquely express 
\[\eta = \sum_{i:[r] \to [N]} \eta_i \otimes e_i \quad \& \quad \zeta = \sum_{i:[r]\to [N]} e_i \otimes \zeta_i,  \]  
where $\eta_i \in H_{l-r},$ $\zeta_i \in H_{m-r}$, and $\{e_i = e_{i(1)} \otimes \ldots e_{i(r)}\}_{i:[r]\to [N]}$ is the standard orthonormal basis for $H_1^{\otimes r}$.  (Here $[x] = \{1,2,\ldots, x\}$.) Recalling formula \eqref{alpha}, we have  \[\alpha_{k}^{l,m} = \Big(\frac{[k+1]_{q}}{\theta_{q}(k,l,m)}\Big)^{1/2}(p_l \otimes p_m)\Big(\iota_{H_{l-r}} \otimes T_r  \otimes \iota_{m-r}\Big)p_k.\]  Noting that $T_r =  \sum_{i:[r] \to [N]}e_i \otimes e_{\check{i}}$, where $\check{i}:[r]\to [N]$ is defined by $\check{i}(s) = i(r-s+1)$.  Using this formula, we then obtain 
\begin{align*}(\alpha_k^{l,m})^*(\eta \otimes \zeta) &=  \Big(\frac{[k+1]_{q}}{\theta_{q}(k,l,m)}\Big)^{1/2} p_k\Big(\sum_{i,j:[r] \to [N]}  T_r^*(e_i \otimes e_j)  \eta_i \otimes \zeta_j  \Big) \\
&=   \Big(\frac{[k+1]_{q}}{\theta_{q}(k,l,m)}\Big)^{1/2}  p_k\Big(\sum_{i:[r] \to [N]} \eta_i \otimes  \zeta_{\check{i}}\Big).
\end{align*}
Applying the triangle and Cauchy-Schwarz inequality to the above expression, we then obtain
\begin{align*}
\|(\alpha_k^{l,m})^*(\eta \otimes \zeta)\|& \le  \Big(\frac{[k+1]_{q}}{\theta_{q}(k,l,m)}\Big)^{1/2} \Big(\sum_{i} \|\eta_i\|^2\Big)^{1/2}\Big(\sum_{i}\|\zeta_{\check{i}}\|\Big)^{1/2}\\
& =\Big(\frac{[k+1]_{q}}{\theta_{q}(k,l,m)}\Big)^{1/2}\|\eta\|_{H_l}\|\zeta\|_{H_m},
\end{align*} 
giving the first inequality.
Next, we observe that with $k = l+m - 2r$, we have from \eqref{theta},
\begin{align*}
\frac{[k+1]_{q}}{\theta_{q}(k,l,m)} &=\frac{[l]_q![m]_q![k +1]_q!}{[r]_q![l-r]_q![m-r]_q![k+r+1]_q!} \\
&=\frac{1}{[r+1]_q}\prod_{s=1}^r \frac{[1+s]_q[l-r+s]_q[m-r+s]_q}{[k+1+s]_q[s]_q^2} \\
&= \frac{1}{[r+1]_q} \prod_{s=1}^r \frac{(1-q^{2+2s})(1-q^{2l-2r+2s})(1-q^{2m-2r+2s})}{(1-q^{2k+2+2s})(1-q^{2s})^2} \\
& \le \frac{1}{[r+1]_q} \Big(\prod_{s=1}^r \frac{1}{1-q^{2s}}\Big)^3 \le \frac{1}{[r+1]_q} \Big(\prod_{s=1}^\infty \frac{1}{1-q^{2s}}\Big)^3. 
\end{align*}
To obtain the second inequality, we just observe that 
\begin{align} \label{ineq-qd}q^r \le \frac{1}{[r+1]_q} = \frac{q^r(1-q^2)}{1-q^{2r+2}} \le q^r(1-q^2)^{-1},  
\end{align}
which gives \[|\langle \alpha_k^{l,m}(\xi)|\eta \otimes \zeta \rangle| \le q^{r/2}(1-q^2)^{-1/2}\Big(\prod_{s=1}^\infty \frac{1}{1-q^{2s}}\Big)^{3/2} = C(q)q^{\frac{l+m-k}{4}}. \]
\end{proof}

Proposition \ref{thm:RD} can be interpreted as giving a general upper bound on the largest Schmidt coefficient of a unit vector belonging to the subspace $\alpha_k^{l,m}(H_k) \subseteq H_l \otimes H_m$.  That is, if $\xi \in H_k$ is a unit vector and $\alpha_k^{l,m}(\xi)$ is represented by its singular value decomposition 
\[
\alpha_k^{l,m}(\xi) = \sum_{i} \sqrt{\lambda_i} e_i \otimes f_i,
\]
with $(e_i)_i\subset H_l, (f_i)_i\subset H_m$ orthonormal systems, and $\lambda_i \ge 0$ satisfy $\sum_i \lambda_i = 1$, then 
\begin{align}\label{SC-bound}
\max_i \lambda_i \le C(q)^2q^{\frac{l+m-k}{2}}.
\end{align}
Since the above quantity is much smaller than $1$ when $k <  l+m$, we conclude that $\alpha_k^{l,m}(H_k)$ is ``far'' from containing containing separable unit vectors of the form $\eta \otimes \zeta \in H_l \otimes H_m$.  That is, $\alpha_k^{l,m} \subset H_l \otimes H_m$ is highly entangled.    We summarize this in the following theorem.

\begin{theorem} \label{highest-weight}
For $k,l,m$ as above, the subspaces $\alpha_k^{l,m}(H_k) \subseteq H_l \otimes H_m$ are (highly) entangled provided $k < l+m$.  When $k=l+m$, the highest weight subspace $\alpha_{l+m}(H_{l+m}) \subset H_l \otimes H_k$ is a separable subspace.  
\end{theorem}

\begin{proof}
The first statement follows immediately from the estimate given by Proposition \ref{thm:RD}.  For the second statement, we exhibit an example of a separable vector in $\alpha_{l+m}^{l,m}(H_{l+m})$.  Observe that if $(e_i)_i$ is an orthonormal basis for $H_1 = \C^N$, then for any $i\ne j$, the elementary tensor $\eta_r(i,j) = e_i \otimes e_j \otimes e_i \otimes \ldots \in H_1^{\otimes r}$ actually belongs to $H_r$.  Indeed, a simple inductive application of the Wenzl recursion formula shows that $p_r\eta_r(i,j) = \eta_r(i,j)$.  Consider now the tensor product vector $\eta_l(i,j) \otimes \eta_m \in H_l \otimes H_m$, where $\eta_m = \eta_m(i,j)$ if $l$ is even and $\eta_m = \eta_m(j,i)$ if $l$ is odd.  Then $\eta_{l}(i,j) \otimes \eta_m = \eta_{l+m}(i,j) = \alpha_{l+m}^{l,m}(\eta_{l+m}(i,j)) \in \alpha_{l+m}^{l,m} (H_{l+m})$ is separable. 
\end{proof}

\subsection{Optimal entanglement estimates and a higher rank rapid decay inequality}

In this section we will investigate to what extend the Schmidt coefficient upper bound \eqref{SC-bound} given by  Proposition \ref{thm:RD} is optimal. It turns out that this bound is in fact optimal in a very strong sense:  For any $d \in \N$, we can find a unit vector $\xi \in H_k$ (provided $N$ is sufficiently large) with the property that  $\alpha_k^{l,m}(\xi)$ admits at least $d$ Schmidt coefficients with the same magnitude as that predicted by \eqref{SC-bound}.  To obtain this optimality result, we first consider a higher rank version of the rapid decay inequality of Proposition \ref{thm:RD}. In the following, we show that if we replace the rapid decay inequality by its ``na\"ive'' extension resulting from the triangle inequality, the resulting upper bound is optimal -- at least for $N$ sufficiently large.

\begin{proposition} \label{higher-rank}
Fix $N \ge 3$ and an admissible triple $(k,l,m) \in \N_0^3$.  For any finite sequences $(\eta_i)_{i=1}^d \subset H_l$ and $(\zeta_i)_{i=1}^d \subset  H_m$, we have
\[\Big\|(\alpha_k^{l,m})^*\Big(\sum_{i=1}^d \eta_i \otimes \zeta_i\Big)\Big\| \le C(q) q^{\frac{l+m-k}{4}} \sum_{i=1}^d \|\eta_i\|\|\zeta_i\|.\]
Moreover, if $d \le (N-2)(N-1)^{\frac{l+m-k-2}{2}}$, then there exist orthonormal systems $(\eta_i)_{i=1}^d \subset H_l$ and $(\zeta_i)_{i=1}^d \subset H_m$ such that   
\[\Big\|(\alpha_k^{l,m})^*\Big(\sum_{i=1}^d \eta_i \otimes \zeta_i\Big)\Big\| = d \Big(\frac{[k+1]_{q}}{\theta_{q}(k,l,m)}\Big)^{1/2} \ge dq^{\frac{l+m-k}{4}}.\]
\end{proposition}

\begin{proof}
The first inequality is simply an application of the triangle inequality to Proposition \ref{thm:RD}. We now consider the second inequality.  As in the proof of Theorem \ref{highest-weight}, consider the vector \[\eta_k(1,2) = \underbrace{e_1\otimes e_2 \otimes e_1 \otimes e_2 \otimes \ldots}_{k \text{ alternating factors}} \in H_k \subset H_1^{\otimes k}.\]  We will also write $\eta_k(1,2) = \eta_0 \otimes \zeta_0$, where $\eta_0 \in H_1^{\otimes l-r}$, $\zeta_0 \in H_1^{\otimes m-r}$ and  $k = l+m-2r$.  Now denote by $A$ the set  of functions $i:[r] \to [N]$ with the property that $i(1) \ge 3$ and with the alternating value condition $i(s) \ne i(s+1)$ for $1 \le s \le r$.  For $i \in A$, define $\eta_i \in H_1^{\otimes l}$ and $\zeta_i \in H_1^{\otimes m}$ by \[\eta_i = \eta_0 \otimes e_{i(1)} \otimes e_{i(2)} \otimes \ldots \otimes e_{i(r)} \quad \& \quad \zeta_i =  e_{i(r)} \otimes \ldots \otimes e_{i(2)} \otimes e_{i(1)}\otimes \zeta_0. \]  Note that the  families $(\eta_i)_{i\in A}$ and $(\zeta_i)_{i\in A}$ are each orthonormal systems of size $(N-2)(N-1)^{r-1}$.    As in the case of the proof of Theorem \ref{highest-weight}, one readily sees from the structure of the Jones-Wenzl projections that $\eta_i = p_l  \eta_i \in H_l$ and $\zeta_i = p_l \zeta_i  \in H_m$, which gives, for each $i \in A$,  

\begin{align*}
(\alpha_k^{l,m})^*(\eta_i \otimes \zeta_i) &= \Big(\frac{[k+1]_{q}}{\theta_{q}(k,l,m)}\Big)^{1/2} p_k\Big(\iota_{H_{l-r}} \otimes T_r^*  \otimes \iota_{m-r}\Big)(\eta_i \otimes \zeta_i) \\ &= \Big(\frac{[k+1]_{q}}{\theta_{q}(k,l,m)}\Big)^{1/2} \eta_k(1,2) \\
\implies \Big\|(\alpha_k^{l,m})^*\Big(\sum_{i \in A} \eta_i \otimes \zeta_i\Big)\Big\| &=
d \Big(\frac{[k+1]_{q}}{\theta_{q}(k,l,m)}\Big)^{1/2}.
\end{align*}
We now conclude by observing that from Remark \ref{upbound}, we have 
\[
\Big(\frac{[k+1]_{q}[r+1]_q}{\theta_{q}(k,l,m)}\Big)^{1/2} \ge 1,
\]
which yields together with inequality \eqref{ineq-qd}
\[
\Big\|(\alpha_k^{l,m})^*\Big(\sum_{i \in A} \eta_i \otimes \zeta_i\Big)\Big\| \ge d[r+1]_q^{-1/2} \ge dq^{\frac{l+m-k}{4}}.
\]
\end{proof}
From the proof of the above higher rank rapid decay inequality, we obtain the optimality result for the Schmidt coefficient bounds alluded to above.  

\begin{theorem}\label{ent-saturation}
Let $(k,l,m) \in \N_0^3$ be an admissible triple, $N \ge 3$, and $d \le (N-2)(N-1)^{\frac{l+m-k-2}{2}}$.  Then there exists a unit vector $\xi \in H_k$ such that $\alpha_k^{l,m}(\xi)$ has a singular value decomposition $\alpha_k^{l,m}(\xi) = \sum_i \sqrt{\lambda_i}e_i \otimes f_i$ with $\lambda_1 \ge \lambda_2 \ge \ldots $ satisfying  \[\lambda_1 = \lambda_2 = \ldots =\lambda_d = \frac{[k+1]_{q}}{\theta_{q}(k,l,m)}  \ge q^{\frac{l+m-k}{2}}.\]
\end{theorem}

\begin{proof}
We first observe that from Proposition \ref{thm:RD}, $\lambda_1 \le  \frac{[k+1]_{q}}{\theta_{q}(k,l,m)}$ for each unit vector $\xi \in H_k$.  To show that this bound is obtained as claimed above, we freely use the notation of Proposition \ref{higher-rank} and its proof.  Let $r= \frac{l+m-k}{2}$ and take $\xi = \eta_k(1,2)=\eta_0 \otimes \zeta_0 \in H_k$ as defined above.  Recalling the definitions of the vectors $(\eta_i)_{i \in A}$ and $(\zeta_i)_{i \in A}$, we then have  
\begin{align*}\alpha_k^{l,m}(\xi) &= \Big(\frac{[k+1]_{q}}{\theta_{q}(k,l,m)}\Big)^{1/2} (p_l \otimes p_m)\Big(\iota_{H_{l-r}} \otimes T_r  \otimes \iota_{m-r}\Big)\xi \\
&=\sum_{i:[r] \to [n]}\Big(\frac{[k+1]_{q}}{\theta_{q}(k,l,m)}\Big)^{1/2} (p_l \otimes p_m)\big(\eta_0 \otimes e_i \otimes e_{\check{i}}\otimes \zeta_0\big) \\
&=\sum_{i \in A} \Big(\frac{[k+1]_{q}}{\theta_{q}(k,l,m)}\Big)^{1/2} \eta_i \otimes \zeta_i +  \sum_{\substack{i:[r] \to [n]\\ i \notin A}}\Big(\frac{[k+1]_{q}}{\theta_{q}(k,l,m)}\Big)^{1/2} (p_l \otimes p_m)\big(\eta_0 \otimes e_i \otimes e_{\check{i}}\otimes \zeta_0\big) \\
\end{align*}
Performing a singular value decomposition on the second sum above, we obtain orthonormal systems $(e_j)_j \subset \text{span}(p_l(\eta_0 \otimes e_i))_{i \notin A} \subset H_l$ and  $(f_j)_j \subset \text{span}(p_m(e_{\check{i}} \otimes \zeta_0))_{i \notin A} \subset H_m$ and Schmidt coefficients $\lambda_j \ge 0$, so that 
\[\alpha_k^{l,m}(\xi) = \sum_{i \in A} \Big(\frac{[k+1]_{q}}{\theta_{q}(k,l,m)}\Big)^{1/2} \eta_i \otimes \zeta_i +  \sum_j \sqrt{\lambda_j} e_j \otimes f_j. \]
The theorem will now follow from the above expression once we observe that this expression is precisely the singular value decomposition of $\alpha_k^{l,m}(\xi)$.  This latter fact is evident because, by construction, $\eta_i \perp p_l(\eta_0 \otimes e_j)$ and $\zeta_i \perp p_m(e_{\check{j}}\otimes \zeta_0)$ for each $i \in A$ and $j \notin A$.
\end{proof}

\begin{remark} \label{rem-optimal}
The number $|A| = (N-2)(N-1)^{\frac{m+l-k-2}{2}}$ of maximal Schmidt coefficients $\lambda_{\max} = \frac{[k+1]_{q}}{\theta_{q}(k,l,m)}$ obtained in Theorem \ref{ent-saturation} is asymptotically maximal in the sense that
\[
\lim_{N \to \infty} |A| \frac{[k+1]_{q}}{\theta_{q}(k,l,m)} = 1.
\]
This shows that in the limit as $N \to \infty$, the vector $\xi \in H_k$ constructed in Theorem \ref{ent-saturation} becomes maximally entangled, with the bulk of its Schmidt coefficients equaling the maximal value $\lambda_{\max}$ allowed by Proposition \ref{thm:RD}.
\end{remark}

\section{Applications} \label{ap}

In this section we consider some applications of the entanglement results of the preceding section to study the outputs of the canonical quantum channels related to our subspaces.  We also construct some new examples of $d$-positive, but not completely positive maps on matrix algebras.

\subsection{$O_N^+$-equivariant quantum channels} \label{app}

Following Section \ref{prelim}, we form, for any admissible triple $(k,l,m) \in \N_0^3$, the complementary pair of quantum channels \[\Phi_k^{\overline{l}, m}: \mc B(H_k) \to \mc B(H_m); \qquad \rho \mapsto (\text{Tr} \otimes \iota)(\alpha_{k}^{l,m} \rho (\alpha_{k}^{l,m})^*),\]
\[\Phi_k^{l, \overline{m}}: \mc B(H_k) \to \mc B(H_l); \qquad \rho \mapsto (\iota \otimes \text{Tr})(\alpha_{k}^{l,m} \rho (\alpha_{k}^{l,m})^*)\]

We then have the following proposition concerning the  $\mc S_1 \to \mc S_\infty$ behavior of these channels. 

\begin{proposition} \label{prop:normbound}
Given any admissible triple $(k,l,m) \in \N_0^3$ and $N \ge 3$,  we have
\begin{align*}&\|\Phi_k^{\overline l,m}\|_{\mc S_1(H_k) \to \mc B(H_m)}=\|\Phi_k^{l, \overline m
}\|_{\mc S_1(H_k) \to \mc B(H_l)} \\
&=  \frac{[k+1]_q}{\theta_q(k,l,m)} \in \big[q^{\frac{l+m-k}{2}}, C(q)^2q^{\frac{l+m-k}{2}}\big].
\end{align*}
\end{proposition}

\begin{proof}
We shall only consider $\Phi_{k}^{\overline l, m}$ as the proof of the other case is identical.   To prove the upper bound $\|\Phi_k^{\overline l, m
}\|_{\mc S_1(H_k) \to \mc B(H_m)} \le  \frac{[k+1]_q}{\theta_q(k,l,m)}$, note that by  completely positivity, convexity and the triangle inequality, it suffices to consider a pure state $\rho =  | \xi \rangle \langle \xi| \in \mc D(H_k)$ and show that $\|\Phi_k^{\overline l, m
}(\rho)\|_{\mc B(H_m)} \le  \frac{[k+1]_q}{\theta_q(k,l,m)}$.  But in this case, we have 
\begin{align*}
\Phi_k^{\overline l, m }(\rho) = (\text{Tr} \otimes \iota)(| \alpha_k^{l,m}\xi \rangle \langle  \alpha_k^{l,m}\xi|) = \sum_i \lambda_i  | f_i \rangle \langle f_i|,
\end{align*}
where $\alpha_k^{l,m}(\xi) = \sum_{i} \sqrt{\lambda_i} e_i \otimes f_i$ is the corresponding singular value decomposition.  In particular, $\|\Phi_k^{\overline l, m
}(\rho)\|_{\mc B(H_m)} = \max_i \lambda_i$, which by Proposition \ref{thm:RD} is bounded above by $\frac{[k+1]_q}{\theta_q(k,l,m)}$.  This upper bound is obtained by taking $\rho = |\xi \rangle \langle \xi |$, where $\xi$ satisfies the hypotheses of Theorem \ref{ent-saturation}.
\end{proof}

The preceeding norm computation for the channels $\Phi_{k}^{\overline l, m}, \Phi_{k}^{ l, \overline m}$ allows for an easy estimate of a  lower bound on their minimum output entropies.

\begin{corollary}\label{MOE-bounds}
Given any admissible triple $(k,l,m) \in \N_0^3$ and $N \ge 3$, we have
\[H_{\min}(\Phi_k^{\overline l, m}), H_{\min}(\Phi_k^{l, \overline m}) \geq \log\Big(\frac{\theta_q(k,l,m)}{[k+1]_q}\Big) \ge -\Big(\frac{l+m-k}{2}\Big) \log(q) - 2\log (C(q)).\]
\end{corollary}

\begin{proof}Given a quantum channel $\Phi: \mc B(H) \to \mc B(K)$ and $\rho \in \mc D(H)$, we note that $H(\Phi(\rho))  = -\sum_i  \lambda_i \log \lambda_i$, where $(\lambda_i)_i$ is the spectrum of $\Phi(\rho)$.  In particular, we have the estimate 
\[
H(\Phi(\rho)) \ge -\log \Big( \max_i\lambda_i  \Big) = -\log\|\Phi(\rho)\|_{\mc B(K)} \ge  -\log\|\Phi\|_{\mc S_1(H) \to \mc B(K)}.\] 
The first inequality in the corollary now follows immediately from Proposition \ref{prop:normbound}.  The second inequality is just a consequence of the inequality $\frac{[k+1]_q}{\theta_q(k,l,m)} \le C(q)^2q^{\frac{l+m-k}{2}}$.
\end{proof}

\begin{remark}
The above estimates show that for $N$ large and $k < l+m$ fixed, the  minimum output entropy of the channels is quite large and grows logarithmically in $N$.     These estimates stand in contrast to what happens in the case of the $SU(2)$-equivariant quantum channels studied by Al Nuwairan in \cite[Section 3]{Al13}. 

In the case where $k = l+m$ (the highest weight case),  we note that  \[H_{\min}(\Phi_k^{\overline l, m}) =  H_{\min}(\Phi_k^{l, \overline m}) = 0, \]
which follows from the fact that $\alpha_{k}^{l,m}(H_{k}) \subseteq H_l \otimes H_m$ is a separable subspace (cf. Theorem \ref{highest-weight}).  Indeed, in this case, there exist unit vectors $\xi \in H_k$, $\eta \in H_l$ and $\zeta \in H_m$ with $\alpha_{k}^{l,m}(\xi) = \eta \otimes \zeta \in H_l \otimes H_m$.  We then have \begin{align*}
H_{\min}(\Phi_k^{\overline l, m}) &\le H(\Phi_k^{\overline l, m}(|\xi\rangle \langle \xi |))= H(|\zeta\rangle \langle \zeta |))  = 0, \\
H_{\min}(\Phi_k^{ l, \overline m}) &\le H(\Phi_k^{l, \overline m}(|\xi\rangle \langle \xi |))= H(|\eta\rangle \langle \eta |))  = 0.
\end{align*}
\end{remark}

\begin{remark}
We expect that the lower bound for the minimum output entropies given in Corollary \ref{MOE-bounds} to be asymptotically optimal as $N \to \infty$, 
at least in some cases (e.g. $m$ fixed).  Evidence for this is provided by Theorem \ref{ent-saturation} and Remark \ref{rem-optimal}, which shows that $\alpha_k^{l,m}(H_k)$ contains unit vectors which are asymptotically maximally entangled with the bulk of their Schmidt coefficients equal to $\frac{[k+1]_q}{\theta_q(k,l,m)}$.  

On the other hand, we can not exclude that the bound is not tight in general, and if the bound could be significantly improved, it could have potential applications towards the problem of constructing a non-random example of a quantum channel that is not MOE additive.   At this stage, however, we are unable to complete the proof that the bound is tight in full generality and leave it as an open question.  
\end{remark}

\begin{remark}
Under some asymptotic regimes (typically, $m$ fixed, $k,N\to\infty$ independently, and any $l$ that makes the triple
$k,l,m$ admissible), 
one can use techniques developed in \cite{BeCoNe12} to describe exactly the image of pure states in $\mc D(H_k)$
under $\Phi_k^{\overline l, m}$. We leave these investigations for future work.
\end{remark}

\subsection{Positive maps} \label{positive}
In this final section we indicate how our representation theoretic model for highly entangled subspaces can be used to construct non-random examples of $d$-positive maps between matrix algebras that are not completely positive. 

First we recall the definition of the Choi matrix of a linear map $\Phi: \mc B(H_A)\to \mc B(H_B)$: it is the matrix $C_\Phi \in \mc B (H_A\otimes H_B)$ 
given by 
$$C_\Phi = \sum_{i,j} e_{ij} \otimes \Phi (e_{ij}),$$
where $e_{ij}$ are the canonical matrix units of $\mc B(H_A)$.
The important property of $C_\Phi$ is that $\Phi$ is completely positive if and only if $C_\Phi$ is positive semidefinite \cite{Ch75}.  Moreover, $C_\Phi$ can be used to detect whether or not $\Phi$ is $d$-positive for any $d \in \N$ \cite{HoLiPoQiSz}:  $\Phi$ is $d$-positive if and only if 
\[\langle C_{\Phi_t} x|x\rangle \ge 0\] 
for all $x  \in H_A \otimes H_B$ with {\it Schmidt rank} bounded above by $d$.  (That is, $x$ admits a singular value decomposition $x = \sum_{i=1}^s \sqrt{\lambda_i} e_i \otimes f_i$ with $\min_i \lambda_i > 0$ and $s \le d$).

Let us now return our usual setup of an admissible triple $(k,l,m) \in \N_0^3$ corresponding to a non-highest-weight inclusion $\alpha_k^{l,m}:H_k \hookrightarrow H_l\otimes H_m$ of irreducible representations of $O_N^+$, $N \ge 3$.  For each $t\ge 0$, we consider the linear map $\Phi_t:  \mc B(H_l)\to \mc B (H_m)$ whose Choi matrix is given by $$C_{\Phi_t} = \iota_{H_l\otimes H_m}-t\alpha_{k}^{l,m}(\alpha_{k}^{l,m})^*.$$

Evidently ${\Phi_t}$ is completely positive if and only if $t \le 1$.  On the other hand, we can prove the following result.
\begin{theorem}
For each $d \in \N$, the map $\Phi_t: \mc B(H_l) \to \mc B(H_m)$ is $d$-positive (but not completely positive) if and only if \[1 < t \le \frac{\theta_{q}(k,l,m)}{d[k+1]_{q}} \le  C(q)^{-2}q^{-\frac{l+m-k}{2}}d^{-1}.\] 
\end{theorem}

\begin{proof}
We have already observed that $\Phi_t$ is not completely positive when $t >1$.  Now fix $d \in \N$ and $x = \sum_{i=1}^s \sqrt{\lambda_i} e_i \otimes f_i \in H_l \otimes H_m$ with Schmidt-rank at most $d$.  Using the first inequality of Proposition \ref{higher-rank} and the Cauchy-Schwarz inequality, we have 
\begin{align*}
\langle C_{\Phi_t} x|x\rangle &= \|x\|^2 - t \langle\alpha_{k}^{l,m}(\alpha_{k}^{l,m})^* (x)|x\rangle \\
&\ge  \|x\|^2 - t  \frac{[k+1]_{q}}{\theta_{q}(k,l,m)} \Big(\sum_{1 \le i \le s} \sqrt{\lambda_i}\Big)^2 \\
&\ge  \|x\|^2 - t  \frac{[k+1]_{q}}{\theta_{q}(k,l,m)} s\|x\|^2 \\
&\ge \|x\|^2\Big(1-t d\frac{[k+1]_{q}}{\theta_{q}(k,l,m)} \Big).
\end{align*}
From this inequality, we obtain $d$-positivity of $\Phi_t$ provided $1-t d\frac{[k+1]_{q}}{\theta_{q}(k,l,m)} \ge 0$, as claimed.

Now assume that $t > \frac{\theta_{q}(k,l,m)}{d[k+1]_{q}}$.  To show that $\Phi_t$ is not $d$-positive, let $x = \sum_{i=1}^d \eta_i \otimes \zeta_i \in H_l \otimes H_m$ be the vector appearing in the second inequality in the statement of Proposition \ref{higher-rank}.  Then the Schmidt-rank of $x$ is $d$ and we have
\begin{align*}
\langle C_{\Phi_t} x|x\rangle &= \|x\|^2 - t \langle\alpha_{k}^{l,m}(\alpha_{k}^{l,m})^* (x)|x\rangle= d -  t  \frac{[k+1]_{q}}{\theta_{q}(k,l,m)} d^2< 0,
\end{align*}  
from which we can conclude.
\end{proof}

\begin{remark}
The above theorem can readily be used to construct maps on matrix algebras that are $d$ positive but not $d+1$ positive.  Indeed, one just has to choose $t>1$, $N \ge 3$ and an admissible triple $(k,l,m) \in \N_0^3$ so that 
\[
\frac{\theta_{q}(k,l,m)}{(d+1)[k+1]_{q}} < t \le \frac{\theta_{q}(k,l,m)}{d[k+1]_{q}}.
\]
Then the corresponding $\Phi_t$ will do the job.
\end{remark}

\bibliographystyle{alpha}
\bibliography{brannan-biblio}

\def\cprime{$'$} \def\cprime{$'$}
\begin{thebibliography}{HLSW04}

\bibitem[AN13]{Al13}
Muneerah Al~Nuwairan.
\newblock The minimal output entropy of {EPOSIC} channels and their {E.B.T.}
  property.
\newblock Preprint, arXiv:1312.2200, 2013.

\bibitem[AN14]{Al14}
Muneerah Al~Nuwairan.
\newblock The extreme points of {SU}(2)-irreducibly covariant channels.
\newblock {\em Internat. J. Math.}, 25(6):1450048, 30, 2014.

\bibitem[ASW11]{AuSzWe11}
Guillaume Aubrun, Stanis{\l}aw Szarek, and Elisabeth Werner.
\newblock Hastings's additivity counterexample via {D}voretzky's theorem.
\newblock {\em Comm. Math. Phys.}, 305(1):85--97, 2011.

\bibitem[ASY14]{AuSzYe14}
Guillaume Aubrun, Stanis{\l}aw~J. Szarek, and Deping Ye.
\newblock Entanglement thresholds for random induced states.
\newblock {\em Comm. Pure Appl. Math.}, 67(1):129--171, 2014.

\bibitem[Ban96]{Ba96}
Teodor Banica.
\newblock Th\'eorie des repr\'esentations du groupe quantique compact libre
  {${\rm O}(n)$}.
\newblock {\em C. R. Acad. Sci. Paris S\'er. I Math.}, 322(3):241--244, 1996.

\bibitem[BC16]{BrCo16}
Michael Brannan and Beno\^{i}t Collins.
\newblock Dual bases in {T}emperley-{L}ieb algebras, quantum groups, and a
  question of {J}ones.
\newblock Preprint, arXiv:1608.03885, 2016.

\bibitem[BCN12]{BeCoNe12}
Serban Belinschi, Beno{\^{\i}}t Collins, and Ion Nechita.
\newblock Eigenvectors and eigenvalues in a random subspace of a tensor
  product.
\newblock {\em Invent. Math.}, 190(3):647--697, 2012.

\bibitem[BDRV06]{BiDeVa06}
Julien Bichon, An~De~Rijdt, and Stefaan Vaes.
\newblock Ergodic coactions with large multiplicity and monoidal equivalence of
  quantum groups.
\newblock {\em Comm. Math. Phys.}, 262(3):703--728, 2006.

\bibitem[Bra16]{Br16}
Michael Brannan.
\newblock Approximation properties for locally compact quantum groups.
\newblock Preprint, arXiv:1605.01770, to appear in the volume {\em Topological
  Quantum Groups}, published by {\em Banach Center Publications, Warszawa},
  2016.

\bibitem[Cho75]{Ch75}
Man~Duen Choi.
\newblock Completely positive linear maps on complex matrices.
\newblock {\em Linear Algebra and Appl.}, 10:285--290, 1975.

\bibitem[FK97]{FrKh97}
Igor~B. Frenkel and Mikhail~G. Khovanov.
\newblock Canonical bases in tensor products and graphical calculus for {$U\sb
  q({ s}{ l}\sb 2)$}.
\newblock {\em Duke Math. J.}, 87(3):409--480, 1997.

\bibitem[FK10]{FuKi10}
Motohisa Fukuda and Christopher King.
\newblock Entanglement of random subspaces via the {H}astings bound.
\newblock {\em J. Math. Phys.}, 51(4):042201, 19, 2010.

\bibitem[GHP10]{GrHoPa10}
Andrzej Grudka, Michal Horodecki, and Lukasz Pankowski.
\newblock Constructive counterexamples to additivity of minimum output
  {R}\'enyi entropy of quantum channels for all $p>2$.
\newblock {\em J. Phys. A: Math. Theor.}, 43, 2010.

\bibitem[Haa79]{Ha78}
Uffe Haagerup.
\newblock An example of a nonnuclear {$C\sp{\ast} $}-algebra, which has the
  metric approximation property.
\newblock {\em Invent. Math.}, 50(3):279--293, 1978/79.

\bibitem[Has09]{Ha09}
Matthew~B. Hastings.
\newblock Superadditivity of communication capacity using entangled inputs.
\newblock {\em Nature Physics}, 5:255--257, 2009.

\bibitem[HLPS12]{HoLiPoQiSz}
J.~Hou, C.-K. Li, X.~Poon, Y.-T.and~Qi, and N.-S. Sze.
\newblock Criteria and new classes of $k$-positive maps.
\newblock Preprint, arXiv:1211.0386, 2012.

\bibitem[HLSW04]{HaLeShWi04}
Patrick Hayden, Debbie Leung, Peter~W. Shor, and Andreas Winter.
\newblock Randomizing quantum states: constructions and applications.
\newblock {\em Comm. Math. Phys.}, 250(2):371--391, 2004.

\bibitem[HLW06]{HaLeWi06}
Patrick Hayden, Debbie~W. Leung, and Andreas Winter.
\newblock Aspects of generic entanglement.
\newblock {\em Comm. Math. Phys.}, 265(1):95--117, 2006.

\bibitem[HW08]{HaWi08}
Patrick Hayden and Andreas Winter.
\newblock Counterexamples to the maximal {$p$}-norm multiplicity conjecture for
  all {$p>1$}.
\newblock {\em Comm. Math. Phys.}, 284(1):263--280, 2008.

\bibitem[KL94]{KaLi94}
Louis~H. Kauffman and S{\'o}stenes~L. Lins.
\newblock {\em Temperley-{L}ieb recoupling theory and invariants of
  {$3$}-manifolds}, volume 134 of {\em Annals of Mathematics Studies}.
\newblock Princeton University Press, Princeton, NJ, 1994.

\bibitem[Mor15]{Mo15}
Scott Morrison.
\newblock A formula for the jones-wenzl projections.
\newblock Preprint, arXiv:1503.00384, 2015.

\bibitem[NC00]{NiCh}
Michael~A. Nielsen and Isaac~L. Chuang.
\newblock {\em Quantum computation and quantum information}.
\newblock Cambridge University Press, Cambridge, 2000.

\bibitem[Tim08]{Ti08}
Thomas Timmermann.
\newblock {\em An invitation to quantum groups and duality}.
\newblock EMS Textbooks in Mathematics. European Mathematical Society (EMS),
  Z\"urich, 2008.
\newblock From Hopf algebras to multiplicative unitaries and beyond.

\bibitem[TL71]{TeLi71}
H.~N.~V. Temperley and E.~H. Lieb.
\newblock Relations between the ``percolation'' and ``colouring'' problem and
  other graph-theoretical problems associated with regular planar lattices:
  some exact results for the ``percolation'' problem.
\newblock {\em Proc. Roy. Soc. London Ser. A}, 322(1549):251--280, 1971.

\bibitem[Ver05]{Ve05}
Roland Vergnioux.
\newblock Orientation of quantum {C}ayley trees and applications.
\newblock {\em J. Reine Angew. Math.}, 580:101--138, 2005.

\bibitem[Ver07]{Ve07}
Roland Vergnioux.
\newblock The property of rapid decay for discrete quantum groups.
\newblock {\em J. Operator Theory}, 57(2):303--324, 2007.

\bibitem[VV07]{VaVe07}
Stefaan Vaes and Roland Vergnioux.
\newblock The boundary of universal discrete quantum groups, exactness, and
  factoriality.
\newblock {\em Duke Math. J.}, 140(1):35--84, 2007.

\bibitem[Wan95]{Wa95}
Shuzhou Wang.
\newblock Free products of compact quantum groups.
\newblock {\em Comm. Math. Phys.}, 167(3):671--692, 1995.

\bibitem[Wen87]{We87}
Hans Wenzl.
\newblock On sequences of projections.
\newblock {\em C. R. Math. Rep. Acad. Sci. Canada}, 9(1):5--9, 1987.

\bibitem[Wor98]{Wo98}
S.~L. Woronowicz.
\newblock Compact quantum groups.
\newblock In {\em Sym\'etries quantiques ({L}es {H}ouches, 1995)}, pages
  845--884. North-Holland, Amsterdam, 1998.

\end{thebibliography}

\end{document}